\documentclass[reqno]{amsart}
\usepackage{amsfonts}
\usepackage{amssymb}
\usepackage{amsmath}
\usepackage{amsaddr}
\usepackage{amsrefs}
\usepackage{tikz-cd}
\usepackage{centernot}
\usepackage{enumerate}
\usepackage{fullpage}
\usepackage{comment}

\theoremstyle{definition}

\theoremstyle{remark}

\newtheorem{thm}{Theorem}[section]
\newtheorem{lem}[thm]{Lemma}
\newtheorem{cor}[thm]{Corollary}
\newtheorem{prop}[thm]{Proposition}

\newtheorem{defn}[thm]{Definition}
\newtheorem{rem}[thm]{Remark}
\newtheorem{ex}[thm]{Example}

\newcommand{\ol}{\overline}
\newcommand{\ts}{\textsc}

\newcommand{\B}{\mathcal B}
\newcommand{\C}{\mathcal C}

\newcommand{\E}{\mathcal E}

\renewcommand{\H}{\mathcal H}

\newcommand{\K}{\mathcal K}

\newcommand{\KK}{\mathbb K}

\newcommand{\NN}{\mathbb N}

\newcommand{\RR}{\mathbb R}

\newcommand{\ZZ}{\mathbb Z}

\newcommand{\mc}{\mathcal}

\renewcommand{\:}{\colon}
\newcommand{\subsetof}{\subseteq}
\newcommand{\<}{\langle}
\renewcommand{\>}{\rangle}

\newcommand{\covers}{<\!\!\!\cdot\,\,}

\newcommand{\cov}{<\!\!\!\cdot\,\,}

\renewcommand{\:}{\colon}
\renewcommand{\>}{\rangle}
\newcommand{\type}{\mathrm{type}}

\allowdisplaybreaks

\begin{document}

\title[]%
{Completely hereditarily atomic \ts{oml}s}
\author[John Harding \and Andre Kornell]%
{John Harding \and Andre Kornell}

\newcommand{\acr}{\newline\indent}

\thanks{The first listed author gratefully acknowledges support of US Army grant W911NF-21-1-0247 and NSF grant DMS-2231414. The second author gratefully acknowledges support of Air Force Office of Scientific Research grant FA9550-21-1-0041.}


\begin{abstract}
An irreducible complete atomic \ts{oml} of infinite height cannot be algebraic and have the covering property. However, modest departure from these conditions allows infinite-height examples. We use an extension of Kalmbach's construction to the setting of infinite chains to provide an example of an infinite-height, irreducible, algebraic \ts{oml} with the $2$-covering property, and Keller's construction provides an example of  an infinite-height, irreducible, complete \ts{oml} that has the covering property and is completely hereditarily atomic. Completely hereditarily atomic \ts{omls} generalize algebraic \ts{oml}s suitably to quantum predicate logic. 
\end{abstract}

\maketitle


\section{Introduction}

A cornerstone of lattice theory is the coordinatization of subspace lattices of Birkhoff and Menger~\cite{BLT}. 
A {\em geomodular lattice} \cite{CrawleyDilworth}*{p.~107} is an 
irreducible, algebraic, complemented modular lattice. These conditions imply that a geomodular lattice is complete and atomic. For a vector space $\mc{V}$ over a division ring, the lattice $\mc{S}(\mc{V})$ of subspaces of $\mc{V}$ is a geomodular lattice, and conversely, each geomodular lattice of height at least four is isomorphic to the subspace lattice of some vector space over a division ring. 

The collection $\C(\H)$ of closed subspaces of a Hilbert space $\mc{H}$ is the prototypical example of an \emph{orthomodular lattice} (\ts{oml}). Explicitly, it is a lattice with respect to the inclusion order, and it is equipped with an orthocomplementation $x \mapsto x^\perp$ that satisfies the \emph{orthomodular law}: $x\leq y\Rightarrow x\vee(x^\perp\wedge y) = y$ \cite{Husimi}. The orthocomplementation maps a closed subspace $x$ to its orthogonal complement $x^\perp$. 

An {\em orthomodular space} $\E$ \cite{GrossKunzi} is a vector space over an involutory division ring that is equipped with an anisotropic Hermitian form whose biorthogonally closed subspaces $\C(\E)$ form an \ts{oml}. There is a well-known characterization theorem for \ts{oml}s of closed subspaces of orthomodular spaces analogous to the characterization theorem for subspace lattices. A {\em propositional system} \cite{GrossKunzi} is a complete, irreducible, atomic \ts{oml} with the {\em covering property}: for any atom $a$ each interval $[x,x\vee a]$ has height at most one. For an orthomodular space $\E$, the closed subspace lattice $\C(\E)$ is a propositional system, and each propositional system of height at least 4 is isomorphic to the closed subspace lattice of an orthomodular space. For a brief history of this result, see \cite{GrossKunzi}*{p.~191}. 

In the passage from geomodular lattices to propositional systems we lose both modularity and algebraicity. In broad terms, the purpose of this note is to consider these properties and their weaker forms in the setting of complete atomic \ts{oml}s. The weaker form of modularity that we consider is the covering property. Our weaker form of algebraicity is the \ts{oml} being \emph{completely hereditarily atomic}: each of its maximal Boolean subalgebras, i.e., its \emph{blocks}, is algebraic. For complete atomic \ts{oml}s this condition is equivalent to requiring that all complete subalgebras are atomic, explaining the name. Every algebraic \ts{oml} is completely hereditarily atomic, but not conversely. 

In the second section of this note, we provide the basics of algebraic \ts{oml}s and of completely hereditarily atomic \ts{oml}s. In particular, it is shown that every algebraic \ts{oml} is atomic and is the direct product of directly irreducible algebraic \ts{oml}s. We further show that no directly irreducible algebraic \ts{oml} of infinite height can have the covering property. It follows that the algebraic \ts{oml}s with the covering property are exactly the products of finite height, orthocomplemented  geomodular lattices. The remaining two sections of the paper are devoted to somewhat complicated examples that show that these results are sharp in that the covering property cannot be weakened to the 2-covering property and algebraicity cannot be weakened to being completely hereditarily atomic.
In the third section we use Kalmbach's construction \cite{KalmbachOriginalConstruction} to produce an irreducible algebraic \ts{oml} of infinite height with the {\em 2-covering property}: for an atom $a$, each interval $[x,x \vee a]$ has height at most two. In the fourth section of this note we use Keller's construction \cite{Keller} to provide a quadratic space $\E$ so that $\C(\E)$ is an infinite height, completely hereditarily atomic, irreducible \ts{oml} with the covering property.

While our results are purely order-theoretic, their motivation lies in quantum predicate logic. One approach to quantum predicate logic uses the quantum sets of noncommutative geometry \cite{Kornellquantumstructures}. A quantum set in this sense is a von Neumann algebra that is {\em hereditarily atomic}: every maximal commutative subalgebra is atomic \cite{Kornellquantumsets}. A von Neumann algebra is hereditarily atomic iff its \ts{oml} of projection operators is completely hereditarily atomic or, equivalently, algebraic. This \ts{oml} of projection operators is then a product of \ts{oml}s of the form $\C(\H)$, where $\H$ is finite-dimensional.

The drawback of this approach to quantum predicate logic is that some infinite-dimensional quantum systems are necessarily modelled by irreducible complete \ts{oml}s of infinite height. We are thus motivated to find a variant of this framework that includes such \ts{oml}s, and we expect them to form a category that is analogous to the category of complete atomic Boolean algebras and tense operators \cite{JonssonTarski} or, equivalently, to the category of sets and relations. Hence, we expect such \ts{oml}s to be completely hereditarily atomic \cite{KornellRel}*{Lemma~4.2}. The \ts{oml} that we obtain using Keller's construction in the fourth section has all of these properties.

\section{Basic observations}

In this section, we give the basic definitions and their consequences. The bounds of a bounded lattice will often be written as $0$ and $1$ without further comment. A \emph{bounded sublattice} is a sublattice that contains these bounds. We now turn to orthomodular lattices, which are structures that generalize Boolean algebras (\ts{ba}s) and the lattice of closed subspaces of a Hilbert space. See \cites{BrunsHarding,KalmbachBook} for further details.

\begin{defn}
An \emph{ortholattice} (\ts{ol}) is a bounded lattice $L$ with an additional order-inverting period-two unary operation $\perp:L\to L$ where $x^\perp$ is a complement of $x$ for each $x\in L$. An \emph{orthomodular lattice} (\ts{oml}) is an \ts{ol} that satisfies the following condition known as {\em the orthomodular law}: $x\leq y\Rightarrow x\vee(x^\perp\wedge y) = y.$
\end{defn}

A maximal Boolean subalgebra of an \ts{ol} $L$ is called a block of $L$ \cite{KalmbachBook}*{p.~19 ff}. For $x\in L$, we have that $\{0,x,x^\perp,1\}$ is a Boolean subalgebra of $L$, and it follows from an application of Zorn's lemma that each $x\in L$ belongs to a block of $L$. So each \ts{ol} is the set-theoretic union of its blocks. \ts{oml}s are characterized \cite{BrunsHarding}*{Prop.~3.2} as those \ts{ol}s where $x\leq y$ implies that $x$ and $y$ belong to a common block of $L$. In other words, \ts{oml}s are those \ts{ol}s whose order structure is determined by the order structures of its blocks. 

\begin{defn}
If $L$ is a bounded lattice, then an element $a\in L$ is an \emph{atom} if $a\neq 0$ and there exists no element $x$ with $0<x<a$.  
\end{defn}

We say that $x$ is {\em below}, {\em beneath}, or {\em under} $y$ if $x \leq y$. A bounded lattice is called \emph{atomic} if each non-zero element has an atom beneath it and is called \emph{atomistic} if each element is the join of the atoms beneath it. An \ts{oml} is atomic iff it is atomistic \cite{KalmbachBook}*{p.~140}. If $L$ is an \ts{oml} and $a$ is an atom of a block of $L$, then $a$ is an atom of $L$ \cite{KalmbachBook}*{p.~39}. 

\begin{defn}
Let $L$ be a complete \ts{oml}. A subalgebra $S$ of $L$ is a \emph{complete subalgebra} if $S$ is closed under all joins and meets as taken in $L$.
\end{defn}

If $S$ is a complete subalgebra of a complete \ts{oml} $L$, then $S$ is itself a complete \ts{oml}. However, if $S$ is a subalgebra of a complete \ts{oml} $L$ such that $S$ is itself complete, then $S$ need not be a complete subalgebra of $L$ since the joins and meets in $S$ might not agree with those in $L$. The following known fact \cite{KalmbachBook}*{p.~24} will be used repeatedly. 

\begin{prop}
If $B$ is a block of an \ts{oml} $L$, then $B$ is closed under all existing joins and meets in $L$. So if $L$ is complete, each of its blocks is a complete subalgebra.    
\end{prop}

\begin{ex}\label{Lebesgue}
A phenomenon of considerable importance occurs with the complete \ts{oml} $\mc{C}(\mc{H})$ of closed subspaces of a Hilbert space $\mc{H}$. This is an atomic \ts{oml} with its atoms being the 1-dimensional subspaces of $\mc{H}$. Each orthonormal basis gives an atomic block of $\mc{C}(\mc{H})$, and all atomic blocks arise this way. However, there are atomless blocks of $\mc{C}(\mc{H})$ as well. For instance, if $\mc{H}$ is realized as $L^2(\mathbb{R})$, then each Lebesgue measurable set $A$ modulo sets of measure zero, gives a closed subspace $S_A$ of functions supported a.e.\ on $A$. The set of all closed subspaces $S_A$ forms a block of $\mc{C}(\mc{H})$ that is isomorphic to the complete Boolean algebra of Lebesgue measurable subsets of $\mathbb{R}$ modulo sets of measure zero. Hence, it is a block of $\mc{C}(\mc{H})$ that has no atoms. It is the existence of such atomless blocks that requires the consideration of purely continuous spectra when diagonalizing self-adjoint operators.    
\end{ex}

We come now to the primary definition of this paper. 

\begin{defn}
An \ts{oml} is \emph{completely hereditarily atomic} if it is complete and each of its complete subalgebras is atomic. 
\end{defn}

\begin{rem}
There is a standard notion of a hereditarily atomic Boolean algebra. This is a Boolean algebra with the property that each of its subalgebras is atomic, or equivalently, that each of its quotients is atomic. This notion was introduced by Mostowski and Tarski and considered by a number of authors. See \cite{HereditarilyAtomicBA}*{Sec.~2} for an account. Our notion of completely hereditarily atomic is formally similar, but even in the Boolean case it is quite different. For instance, as we will soon see, a powerset Boolean algebra is completely hereditarily atomic but it is not hereditarily atomic unless it is finite.
\end{rem}

We next consider the property of being completely hereditarily atomic 
as it applies to Boolean algebras. We begin by recalling some general lattice-theoretic notions that will be pertinent. 
\vspace{1ex}

A complete lattice is \emph{completely distributive} \cite{Compendium}*{p.~59} if $\bigvee_{i \in I}\bigwedge_{j \in J_i} a_{ij} = \bigwedge_{\alpha \in \prod_{i \in I} J_i}\bigvee_{i \in I} a_{i\alpha(i)}$ for any suitably indexed family of elements.
An element $c$ in a complete lattice is \emph{compact} \cite{CrawleyDilworth}*{p.~15} if for any $S\subseteq L$, the inequality $c\leq \bigvee S$ implies that there is a finite subset $S'\subseteq S$ with $c\leq\bigvee S'$. 
A complete lattice is \emph{algebraic} \cite{CrawleyDilworth}*{p.~46}, also known as compactly generated, if every element is the join of the compact elements beneath it. A \emph{cover} in a lattice, denoted $a \covers b$, is a pair of elements with $a<b$ such that there exists no $c$ with $a<c<b$. Let $[x,y]=\{z:x\leq z\leq y\}$ be the usual closed interval. A lattice is \emph{strongly atomic} \cite{CrawleyDilworth}*{p.~14} if each interval $[x,y]$ is atomic and {\em weakly atomic} \cite{CrawleyDilworth}*{p.~14} if for each $x<y$, there is a cover $a\covers b$ in the interval $[x,y]$. We collect some known facts: 

\pagebreak[3]

\begin{prop}\label{nbj}
$ $
\begin{enumerate}
\item The product of algebraic lattices is algebraic. 
\item A complete sublattice of an algebraic lattice is algebraic. 
\item Each algebraic lattice is weakly atomic. 
\end{enumerate}    
\end{prop}

The first two statements in Proposition~\ref{nbj} are found in \cite{Compendium}*{p.~88}, and the third statement is found in \cite{CrawleyDilworth}*{p.~14}. It is an important fact for the theory of continuous lattices that the class of algebraic lattices is not closed under taking images of complete homomorphisms. Easy examples can be found in chains. 

\begin{prop} \label{equiv1}
For $B$ a complete Boolean algebra the following are equivalent: 
\begin{enumerate}
\item $B$ is a powerset Boolean algebra. 
\item $B$ is atomic. 
\item $B$ is completely distributive. 
\item $B$ is algebraic.
\item $B$ is completely hereditarily atomic. 
\item $B$ is weakly atomic.
\item Each maximal chain in $B$ is weakly atomic. 
\end{enumerate}
\end{prop}

\begin{proof}
(1) $\Leftrightarrow$ (2) $\Leftrightarrow$ (3). This is Tarski's Theorem \cite{CrawleyDilworth}*{p.~35}. (3) $\Leftrightarrow$ (4). \cite{Compendium}*{p.~91}. So we have that (1) -- (4) are equivalent. (3) $\Rightarrow$ (5). If $B$ is completely distributive, then so is every complete subalgebra of $B$, which is hence atomic. (5) $\Rightarrow$ (2). Trivial. So (1) -- (5) are equivalent. (4) $\Rightarrow$ (6). Every algebraic lattice is weakly atomic. 
(6) $\Rightarrow$ (2). Let $x$ be non-zero in $B$. Then $[0,x]$ is non-trivial, so by weak atomicity there is a cover $a\covers b$ with $0\leq a\covers b\leq x$. Then $b\wedge a^\perp$ is an atom of $B$ that lies beneath $x$. Thus (1) -- (6) are equivalent. (7) $\Rightarrow$ (6) If $x<y$ then there is a maximal chain containing both $x,y$. By maximality, a cover in this chain is a cover in $B$. So the interval $[x,y]$ contains a cover. (4) $\Rightarrow$ (7). A maximal chain in $B$ is a complete sublattice of $B$ and hence is algebraic, and every algebraic lattice is weakly atomic. 
\end{proof}

We turn our attention to \ts{oml}s. Elements $x$ and $y$ of an \ts{oml} \emph{commute} if they lie in a block. It is known \cite{KalmbachBook}*{p.~39} that every pairwise commuting subset of an \ts{oml} is contained in a block and that any two comparable elements commute. Thus every chain is contained in a block.

\begin{prop} \label{equiv}
For $L$ a complete \ts{oml} the following are equivalent: 
\begin{enumerate}
\item $L$ is completely hereditarily atomic. 
\item Each block of $L$ is atomic. 
\item Each maximal chain in $L$ is weakly atomic. 
\end{enumerate}
\end{prop}

\begin{proof}
(1) $\Rightarrow$ (2). By definition, since a block is a complete subalgebra. (2) $\Rightarrow$ (3). Let $\mc{C}$ be a maximal chain of $L$. Then $\mc{C}$ is contained in some block of $L$ and is a maximal chain in this block. That $\mc{C}$ is weakly atomic follows from Proposition~\ref{equiv1}. (3) $\Rightarrow$ (1). Let $S$ be a complete subalgebra of $L$. We must show that $S$ is atomic. Let $x$ be a non-zero element of $S$. There is a maximal chain $\mc{D}$ of $S$ that contains $x$. Since this chain is maximal in $S$ it is closed under joins and meets in $S$ and hence also closed under joins and meets in $L$. We can extend $\mc{D}$ to a maximal chain $\mc{C}$ of $L$. Since $\mc{C}$ is weakly atomic, the interval $[0,x]$ of $\mc{C}$ contains a cover $a\covers b$. Since $\mc{C}$ is maximal, $a\covers b$ is also a cover in $L$. Since $\mc{D}$ is closed under joins and meets in $L$, there is a largest element $a^-$ in $\mc{D}$ below $a$ and a least element $b^+$ in $\mc{D}$ above~$b$. Note that $b^+$ lies beneath $x$. We then have that $a^-$ is covered by $b^+$ in $\mc{D}$. Indeed, if an element $z$ of $\mc{D}$ is strictly greater than $a^-$ then it does not lie below $a$, so must lie above~$b$, and since it is in $\mc{D}$ it lies above $b^+$. So indeed, $a^-\covers b^+$ is a cover in $\mc{D}$, and as $\mc{D}$ is a maximal chain in $S$, $a^-\covers b^+$ is a cover in $S$. Then since $b^+\leq x$, we have $b^+\wedge (a^-)^\perp$ is an atom of $S$ beneath $x$. 
\end{proof}

\begin{prop}\label{xxz}
Each algebraic \ts{oml} is completely hereditarily atomic. 
\end{prop}

\begin{proof}
Suppose $L$ is an algebraic \ts{oml} and $B$ is a block of $L$. Then $B$ is a complete subalgebra of $L$, so it is in particular a complete sublattice of $L$. By Proposition~\ref{nbj}, $B$ is an algebraic Boolean algebra. So by Propositions~\ref{equiv1} and \ref{equiv}, $L$ is completely hereditarily atomic.     
\end{proof}

Before examining properties of these classes, we require some basic results about \ts{oml}s. To begin, the \emph{horizontal sum} of a family $(L_i)_{i \in I}$ of \ts{oml}s is the \ts{oml} formed by ``gluing'' together the \ts{oml}s together along the bounds 0, 1. See \cite{KalmbachBook}*{p.~306}. We note also that if $x\leq y$ in an \ts{oml} $L$, then the interval $[x,y]$ of $L$ forms an \ts{oml} with the inherited join and meet operations from $L$ and with orthocomplementation given by $z^\#=x\vee(z^\perp\wedge y)$. Further, the intervals $[x,y]$ and $[0,y\wedge x^\perp]$ are isomorphic. See \cite{BrunsHarding}*{Prop.~2.5, Cor.~2.6}. 

\begin{defn}\label{defn: complete homomorphism}
A \emph{complete homomorphism} between complete \ts{oml}s is a map that preserves all joins and meets as well as orthocomplementation.  \end{defn}

The \emph{center} $C(L)$ of an \ts{oml} \cite{KalmbachBook}*{p.~23} is the set of elements that commute with all others, hence those that belong to every block. If $c$ belongs to the center of $L$, then $L$ is isomorphic to $[0,c]\times[0,c^\perp]$ via the mapping $f(x)=(x\wedge c,x\wedge c^\perp)$ and every direct product decomposition of $L$ arises this way. See \cite{BrunsHarding}*{Lem.~2.7}. We provide one more result. 

\begin{lem}\label{factor}
If $L$ and $M$ are complete \ts{oml}s and $f:L\to M$ is a complete homomorphism onto $M$, then there is a central $c\in L$ with $M$ isomorphic to $[0,c]$.
\end{lem}

\begin{proof}
Let $\theta$ be the kernel of $f$, that is, the set of pairs $(x,y)$ with $f(x)=f(y)$. Since $f$ is a complete homomorphism, there is a largest element of $L$ that is mapped by $f$ to 0. By \cite{BrunsHarding}*{Lem.~5.5}, $\theta$ is a factor congruence, meaning that there is a congruence $\theta'$ so that the natural mapping from $L$ to $L/\theta\times L/\theta'$ is an isomorphism. But $L/\theta$ is isomorphic to $M$, and by \cite{BrunsHarding}*{Lem.~2.7}, there is a central element $c$ with $L/\theta$ isomorphic to $[0,c]$.
\end{proof}

\begin{defn}
Let $\sf{ALG}$ be the class of algebraic \ts{oml}s and $\sf{CHA}$ be the class of completely hereditarily atomic \ts{oml}s.
\end{defn}

\begin{prop}\label{sre}
$ $
\begin{enumerate}
\item $\sf{ALG}\subseteq\sf{CHA}$.
\item $\sf{ALG}$ and $\sf{CHA}$ are closed under products.
\item $\sf{ALG}$ and $\sf{CHA}$ are closed under complete subalgebras. 
\item $\sf{ALG}$ and $\sf{CHA}$ are closed under complete homomorphic images.
\item $\sf{ALG}$ and $\sf{CHA}$ are closed under intervals. 
\item $\sf{CHA}$ is closed under horizontal sums. 
\end{enumerate}
Here, a class $\sf{C}$ is closed under intervals if $L\in\sf{C}$ and $[x,y] \subseteq L$ together imply $[x,y]\in\sf{C}$.
\end{prop}

\begin{proof}
(1) This is Proposition~\ref{xxz}. (2)~For $\sf{ALG}$ this follows from Proposition~\ref{nbj}. For $\sf{CHA}$ this follows from Proposition~\ref{equiv} using the fact that the blocks of a product are the product of the blocks of the factors and that the product of atomic Boolean algebras is atomic. 
(3)~For $\sf{ALG}$ this follows from Proposition~\ref{nbj}. For $\sf{CHA}$ this is a trivial consequence of the definition since a complete subalgebra of a complete subalgebra is a complete subalgebra. 
(4)~Suppose $M$ is a complete homomorphic image of $L$. By Lemma~\ref{factor}, $M$ is isomorphic to $[0,c]$ for some central element $c\in L$. If $L$ is algebraic, clearly $[0,c]$ is algebraic since any element in $[0,c]$ that is compact in $L$ is also compact in $[0,c]$. Suppose $L$ is completely hereditarily atomic. If $\mc{C}$ is a maximal chain in $[0,c]$, then taking a maximal chain $\mc{D}$ in $[c,1]$ provides a maximal chain $\mc{C}\cup\mc{D}$ of $L$. By Proposition~\ref{equiv}, $\mc{C}\cup\mc{D}$ is weakly atomic, and this yields that $\mc{C}$ is weakly atomic. So by Proposition~\ref{equiv}, $[0,c]$ is completely hereditarily atomic.
(5)~Suppose $[x,y]$ is an interval of a complete \ts{oml} $L$. Then by \cite{BrunsHarding}*{Prop.~2.4}, the set $S$ of all elements compatible with both $x$ and $y$ is a complete subalgebra of $L$, and by \cite{BrunsHarding}*{Prop.~2.5}, there is an onto homomorphism $f:S\to [x,y]$ given by $f(z)=x\vee(z\wedge y)$. It follows from \cite{BrunsHarding}*{Prop.~2.4} that $f$ is a complete homorphism. So $[x,y]$ is a complete homomorphic image of a complete subalgebra of $L$. The result follows from previously established items.
(6)~Suppose $L$ is the horizontal sum of a family $(L_i)_{i \in I}$. Then a maximal chain in $L$ is a maximal chain in some $L_i$. The result follows from Proposition~\ref{equiv}. 
\end{proof}

\begin{prop} \label{prod}
If $L$ is a 
completely hereditarily atomic \ts{oml},
then its center is a powerset Boolean algebra and $L$ is a direct product of directly irreducible completely hereditarily atomic \ts{oml}s. 
Further, if $L$ is algebraic, then $L$ is a product of directly irreducible algebraic \ts{oml}s.
\end{prop}

\begin{proof}
Let $L\in\sf{CHA}$. By \cite{BrunsHarding}*{Prop.~2.4}, the center $C(L)$ is a complete subalgebra of $L$ and since the center consists of pairwise commuting elements, it is Boolean. Thus, since $L\in\sf{CHA}$, the center is a complete atomic Boolean algebra. Let $A$ be the set of atoms of $C(L)$, and define $\phi:L\to\prod_A[0,a]$ by setting $\phi(x)(a) = \phi_a(x)= x\wedge a$. As we have noted before, each $\phi_a:L\to [0,a]$ is a complete onto homomorphism, and it follows that $\phi$ is a complete homomorphism. For any $x\in L$, by \cite{BrunsHarding}*{Prop.~2.3}, we have $x=x\wedge\bigvee_{a \in A} a = \bigvee_{a \in A} x\wedge a$. This shows that $\phi$ is one-to-one. 
To see that $\phi$ is onto suppose that $x_a\leq a$ for each $a\in A$. Then the set $\{x_a:a\in A\}$ is pairwise orthogonal, hence pairwise commuting, and so is contained in a block $B$ of $L$. As with every block, $B$ contains $C(L)$. Set $x=\bigvee_{a \in A} x_a$ and note that this also belongs to $B$. Then for $b\in A$ we have $\phi_{b}(x) = b\wedge \bigvee_{a \in A}x_a = \bigvee_{a \in A} b\wedge x_a$. Since the elements of $A$ are pairwise orthogonal and $x_a\leq a$, it follows that $\phi_{b}(x)=x_b$. Thus $\phi$ is onto. So $\phi$ is an isomorphism. Since the center of a product is the product of the centers, it follows that each $[0,a]$ has center $\{0,a\}$ and hence \cite{BrunsHarding}*{Lem.~2.7} is directly irreducible. 

We have shown that $L\in \sf{CHA}$ is isomorphic to a direct product of directly irreducible \ts{oml}s that are intervals of $L$. That these intervals are in $\sf{CHA}$, and in $\sf{ALG}$ if $L\in\sf{ALG}$, follows from Proposition~\ref{sre}. 
\end{proof}

A lattice is {\em chain-finite} if each of its chains is finite and is of {\em finite height} if there is a finite upper bound on the cardinality of its chains. The {\em height} of a finite-height lattice is one less than the maximum cardinality of a chain in it. So the 2-element lattice has height~1. In a modular lattice, if one maximal chain is finite, then all are finite and have the same cardinality. Thus, for modular lattices, chain-finite and finite height are equivalent. A lattice is of {\em infinite height} if it is not of finite height. If an \ts{ol} has an infinite orthogonal set, then it has an infinite chain; in an \ts{oml} the converse is true. So in an \ts{oml}, chain-finite is equivalent to all orthogonal sets being finite. 

\begin{ex}
It is easily seen that every chain-finite lattice is complete and that every element of a chain-finite lattice is compact, so each chain-finite lattice is algebraic. Thus, every chain-finite \ts{oml} is algebraic and hence completely hereditarily atomic.
\end{ex}

\begin{ex}
It need not be that the horizontal sum of algebraic \ts{oml}s is algebraic. Consider the horizontal sum $\mc{P}(\mathbb{N})\oplus\mc{P}(\mathbb{N})$ of two copies of the powerset of the natural numbers. The atoms of the first copy join to 1, but any finite subset of these atoms joins to an element that is not equal to $1$ and that belongs to the first copy. So no nonzero element of the second copy can be compact. This provides an example of a completely hereditarily atomic \ts{oml} that is not algebraic.
\end{ex}

\begin{ex}
Let $\H$ be a Hilbert space and $\C(\H)$ be its \ts{oml} of closed subspaces. If $\H$ is finite-dimensional, then $\C(\H)$ is chain-finite and hence algebraic. If $\H$ is infinite-dimensional, then as noted in Example~\ref{Lebesgue}, it has an atomless block, so $\C(\H)$ is not even completely hereditarily atomic let alone algebraic. It is instructive to see directly why $\C(\H)$ fails to be algebraic when $\H$ is of infinite dimension. Take a non-zero vector $v\in\H$. Then there is an orthonormal basis $(u_i)_{i \in I}$ so that $v$ is not a finite linear combination of any of these basis elements. This shows that the atom $\langle v\rangle$ that is the one-dimenional subspace spanned by $v$  is not compact since $\bigvee_{i \in I} \langle u_i\rangle = 1$ but no finite sub-join lies above $\langle v \rangle$.     
\end{ex}

We next consider the combination of algebraicity, modularity, and their weaker forms, in the setting of \ts{oml}s. We begin with relevant terminology. 
A lattice $L$ is {\em semimodular} if for all $a,b\in L$, if $a\wedge b \covers a$ then $b\covers a\vee b$ and is {\em dual semimodular} if the dual condition holds. A lattice $L$ is {\em relatively complemented} if each interval in $L$ is complemented.   

\begin{prop}
In any bounded lattice, 
\begin{enumerate}
\item modular $\Rightarrow$ semimodular and dual semimodular;
\item algebraic, strongly atomic, semimodular, and dual semimodular $\Rightarrow$ modular; 
\item semimodular $\Rightarrow$ the covering property;
\item modular and complemented $\Rightarrow$ relatively complemented. 
\end{enumerate}
\end{prop}

\begin{proof}
The first two statements are found at \cite{CrawleyDilworth}*{p.~24}, the third is a trivial consequence of the definitions, and the fourth is at \cite{CrawleyDilworth}*{p.~31}.
\end{proof}

\begin{prop} \label{swr}
$ $
\begin{enumerate}
\item Every \ts{oml} is relatively complemented. 
\item Within \ts{oml}s, atomic, strongly atomic, weakly atomic, and atomistic coincide. 
\item Within \ts{oml}s, semimodular, dual semimodular, and the covering property coincide. 
\end{enumerate}
\end{prop}

\begin{proof}
(1) An interval in an \ts{oml} is itself an \ts{oml} and hence is complemented. (2)~The equivalence of atomic and atomistic for \ts{oml}s is in \cite{KalmbachBook}*{p.~140}. An interval $[x,y]$ in an \ts{oml} is isomorphic to the interval $[0,y\wedge x^\perp]$, so atomic implies strongly atomic for \ts{oml}s. If an \ts{oml} $L$ is weakly atomic and $x$ is non-zero, then there are $a,b$ with $0\leq a\covers b\leq x$. It follows that $b\wedge a^\perp$ is an atom below $x$. So weakly atomic implies atomic, and since strongly atomic always implies weakly atomic, the three notions are equivalent in \ts{oml}s. (3)~Since any \ts{oml} is self-dual, semimodular and dual semimodular coincide for \ts{oml}s. Semimodularity always implies the covering property. Suppose $L$ is an \ts{oml} with the covering property and $a,b\in L$ with  $a\wedge b \covers a$. Since $[a\wedge b,a]$ is isomorphic to $[0,a\wedge(a^\perp\vee b^\perp)]$, we have that $a\wedge(a^\perp\vee b^\perp)$ is an atom. The covering property gives $b\covers b\vee(a\wedge(a^\perp\vee b^\perp))$, and since $a^\perp\vee b^\perp$ commutes with $a$ and $b$, the Foulis-Holland theorem \cite{KalmbachBook}*{p.~25} gives $b\vee(a\wedge(a^\perp\vee b^\perp))=a\vee b$. So $L$ is semimodular. 
\end{proof}

\begin{rem}
Consider the class $\mc{W}$ of complete, relatively complemented, weakly atomic lattices. It is known \cite{CrawleyDilworth}*{p.~93} that each member of $\mc{W}$ is the product of subdirectly irreducible lattices belonging to $\mc{W}$. It follows that the notions of directly irreducible and subdirectly irreducible coincide for members of $\mc{W}$. The lattice reduct of a completely hereditarily atomic \ts{oml} belongs to $\mc{W}$. Using this fact and the fact that any lattice quotient of an \ts{oml} carries a orthocomplementation making it an \ts{oml} quotient \cite{BrunsHarding}*{Prop.~4.1}, we can obtain an alternative route to Proposition~\ref{prod}.
\end{rem}

We are now ready for our result characterizing algebraic \ts{oml}s with the covering property.

\begin{thm}\label{predolph}
Let $L$ be a directly irreducible, algebraic \ts{oml} with the covering property. Then, $L$ is modular and has finite height.
\end{thm}

\begin{proof}
Since $L$ is algebraic, it is atomic, and so by Proposition~\ref{swr}, it is strongly atomic. Since $L$ has the covering property, by  Proposition~\ref{swr}, it is semimodular and dually semimodular. Any lattice that is algebraic, strongly atomic, semimodular, and dually semimodular is modular \cite{CrawleyDilworth}*{p.~24}. 
Thus $L$ is modular. Since $L$ is an irreducible, algebraic, complemented modular lattice, if it has height at least four, it is isomorphic to the subspace lattice of a vector space. Since $L$ is orthocomplemented, it is self-dual, so by \cite[Ch.~4]{Baer}, $L$ has finite height.
\end{proof}

In the above proof, rather than use Baer's result, we can use Kaplansky's theorem \cite{KalmbachBook}*{p.~179} that any irreducible, complete modular \ts{ol} is a continuous geometry, and hence has a dimension function with range $[0,1]$. Since $L$ is algebraic, it is atomic, and any two atoms $x,y$ are perspective, meaning that here is an atom $z$ with $x\vee z = y\vee z$ \cite{KalmbachBook}*{p.~209}. It follows that any two distinct atoms have the same dimension. Since the dimension function is finite, $L$ cannot have an infinite pairwise orthogonal set of atoms, so $L$ is of finite height. 

\begin{cor}\label{dolph}
Algebraic \ts{oml}s with the covering property are exactly the products of irreducible finite-height modular \ts{ol}s. 
\end{cor}

\begin{proof}
Every finite-height \ts{ol} is algebraic, and every modular lattice has the covering property. So every finite height modular \ts{ol} is algebraic and has the covering property. The class of algebraic \ts{oml}s is closed under products by Proposition~\ref{sre} and the product of lattices with the covering property also has the covering property. This provides one containment. The other containment follows from Proposition~\ref{prod} and Theorem~\ref{predolph}.
\end{proof}

\begin{rem}
An alternative proof of Theorem~\ref{predolph} and hence of Corollary~\ref{dolph} will be given in Section~4.
\end{rem}
   
\section{Kalmbach's construction}

In this section we construct further examples of completely hereditarily atomic \ts{oml}s via a method called the Kalmbach construction. This method was originally proposed by Kalmbach \cite{KalmbachOriginalConstruction} and used to show that there are no special lattice identities holding in the variety of \ts{oml}s. The construction was used in \cite{HardingKalmbachMacNeille} to discuss completions of \ts{oml}s. In \cite{HardingConcrete} it was used to investigate properties of concrete \ts{oml}s, and it was also shown that the Kalmabach construction was adjoint to the forgetful functor from the category of \ts{omp}s to the category of bounded posets. Jen\v{c}a \cite{Jenca} showed that the Eilenberg-Moore algebras of this adjunction are exactly the effect algebras \cite{EffectAlgebras}. We assume familiarity with the Kalmbach construction as in \cites{HardingKalmbachMacNeille,HardingConcrete} but recall several points. 

\begin{defn}
For $L$ a bounded lattice, let $\mc{K}(L)$ be the set of all even-length strictly increasing sequences $x$ in $L$. Let $\ell(x)$ be half of the length of $x$, and denote $x$ by
$$x: x_0<x_1<\cdots<x_{2\ell(x)-2}<x_{2\ell(x)-1}.$$ 
Define a binary relation $\sqsubseteq$ on $\mc{K}(L)$ by setting
$x\sqsubseteq y$ iff for each $i<\ell(x)$ there is $j<\ell(y)$ with $y_{2j}\leq x_{2i}<x_{2i+1}\leq y_{2j+1}$. Finally, define a unary operation $(\,\cdot\,)^\perp$ on $\mc{K}(L)$ by letting $x^\perp$ be the increasing  sequence whose set of terms is the symmetric difference of the set of terms of $x$ and $\{0,1\}$. Thus, $x^\perp$ is formed from $x$ by inserting $0$ into the sequence if it is not in $x$ and removing it if it is in $x$ and doing likewise with $1$.
\end{defn}

The motivation behind this construction is to ``paste together'' the Boolean algebras generated by the bounded chains of $L$. In \cite{BalbesDwinger}*{p.~106} $\mc{B}(C)$ is constructed as the ring of subsets of $C\setminus\{1\}$ generated by half-open intervals $[x,y)$ and this leads directly to the representations in terms of even-length sequences. The following is found in \cites{HardingKalmbachMacNeille}:

\begin{thm}\label{blocks}
Let $L$ be a bounded lattice. Then $\mc{K}(L)$ is an \ts{oml}, and blocks $B \subseteq \mc{K}(L)$ are in bijective correspondence with maximal chains $C \subseteq L$ via $C \mapsto \mc{K}(C)$.
\end{thm}

\begin{prop}\label{cde}
The atoms of $\mc{K}(L)$ are the covers $a \covers b$ in $L$. Further, 
\begin{enumerate}
\item $\mc{K}(L)$ is atomic iff $L$ is weakly atomic,
\item every block of $\mc{K}(L)$ is atomic iff every maximal chain in $L$ is weakly atomic.
\end{enumerate}
\end{prop}

\begin{proof}
Clearly the atoms of $\mc{K}(L)$ are the covers in $L$. An element $x$ lies above the atom $a \covers b$ iff there is an index $i<\ell(x)$ with the interval $[x_{2i},x_{2i+1}]$ containing this cover. So $\mc{K}(L)$ is atomic iff every non-trivial interval in $L$ contains a cover, i.e., iff $L$ is weakly atomic. This proves claim (1). Claim (2) follows immediately from claim (1) and Theorem~\ref{blocks} because every maximal chain $C \subsetof L$ is itself a bounded lattice.
\end{proof}

The following result is a corollary of \cite{HardingKalmbachMacNeille}*{Prop.~2.1}, noting that any complete lattice $L$ satisfies the condition identified there as ($\dagger$). We write $\ol{\mc{K}(L)}$ for the MacNeille completion, or the completion by cuts, of $\mc{K}(L)$. For details on MacNeille completions, see \cite{HardingKalmbachMacNeille}.

\begin{thm}\label{cdf}
If $L$ is a complete lattice, then $\ol{\mc{K}(L)}$ is an \ts{oml}. 
\end{thm}

\begin{ex}
We give an example of a complete lattice $L$ where $\ol{\mc{K}(L)}$ is a complete atomic \ts{oml} that has a block with no atoms. We first construct a complete lattice $L$ where every non-trivial interval contains a cover, yet has a maximal chain with no covers. Let $[0,1]$ be the real unit interval and set $L=[0,1]\cup\{x_{\lambda_1,\lambda_2}: 0\leq\lambda_1<\lambda_2\leq 1\}$. Define a relation $\leq$ on $L$ to be the union of the linear ordering on $[0,1]$, the identity relation on $L$, the sets $\{(\lambda,x_{\lambda_1,\lambda_2}):\lambda\leq\lambda_1\}$ and $\{(x_{\lambda_1,\lambda_2},\lambda):\lambda_2\leq\lambda\}$ for $\lambda\in[0,1]$, and the set $\{(x_{\lambda_1, \lambda_2}, x_{\lambda_1', \lambda_2'}) : \lambda_2 \leq \lambda_1'\}$. It is easy to verify that this is a partial ordering on $L$ and that $L$ is a complete lattice. The chain $[0,1]$ is a maximal chain in $L$ without covers, and as each $\lambda_1\covers x_{\lambda_1,\lambda_2}$ and $x_{\lambda_1,\lambda_2}\covers \lambda_2$ are covers, it follows that every non-trivial interval in $L$ contains a cover.

By Proposition~\ref{cde}, $\mc{K}(L)$ is atomic but has a block $B$ without atoms. By Theorem~\ref{cdf}, $\ol{\mc{K}(L)}$ is a complete \ts{oml}. It is atomic because $\mc{K}(L)$ is atomic and MacNeille completions are join dense. The block $B$ of $\mc{K}(L)$ is a Boolean subalgebra of $\ol{\mc{K}(L)}$ and hence extends to a block $B'$ of $\ol{\mc{K}(L)}$. An atom of $B'$ is an atom of $\ol{\mc{K}(L)}$ and hence is an atom of $\mc{K}(L)$ since MacNeille completions are join dense. But there are no atoms of $\mc{K}(L)$ that commute with all elements of $B$, so there are none that commute with all elements of $B'$. Thus $B'$ is an atomless block.
\end{ex}

We next use $\ol{\mc{K}(L)}$ to produce examples of completely hereditarily atomic \ts{oml}s with various additional features. We restrict our attention to working with bounded lattices $L$ where every maximal chain is finite or has order type $\omega+1$, i.e. it is isomorphic to $\mathbb{N}\cup\{\infty\}$. 

\begin{defn}
An \emph{$(\omega+1)$-lattice} is a lattice $L$ where every maximal chain is either finite or has order type $\omega+1$. A strictly increasing sequence $x$ in $L$ is \emph{admissible} if it is either finite and of even length or is infinite. Let $\mc{K}^*(L)$ be the set of admissible sequences in $L$. 
\end{defn}

For $x\in\mc{K}^*(L)$ we let $\ell(x)$ be half the length of $x$ if $x$ is finite and be $\omega$ if $x$ is infinite. Then, for any admissible sequence $x$, we have that $\{x_{2k},x_{2k+1}:k<\ell(x)\}$ is the set of its elements. 

\begin{defn}
If $L$ is an $(\omega+1)$-lattice, define a binary relation $\sqsubseteq$ on $\mc{K}^*(L)$ by setting $x\sqsubseteq y$ if for each $i<\ell(x)$ there is $j<\ell(y)$ with $y_{2j}\leq x_{2i}<x_{2i+1}\leq y_{2j+1}$. For $x\in\mc{K}^*(L)$, define $x^\perp$ as before if $\ell(x)$ is finite. If $\ell(x)=\omega$ set 
\[ x^\perp = \begin{cases} x_1<x_1<\cdots&\mbox{if }x_0=0\\ \,0\,<\, x_0<\cdots&\mbox{if }x_0\neq 0 \end{cases} \]
We let $0$ and $1$ be as in $\mc{K}(L)$, the empty sequence and the sequence $0<1$, respectively. 
\end{defn}

\pagebreak[3]

\begin{lem}\label{dvd}
Let $L$ be an $(\omega+1)$-lattice. Then, 
\begin{enumerate}
\item $\mc{K}^*(L)$ is a poset under $\sqsubseteq$ with bounds $0$ and $1$,
\item $\mc{K}(L)$ is a subposet of $\mc{K}^*(L)$, 
\item the join of any two elements in $\mc{K}(L)$ is their join in $\mc{K}^*(L)$. 
\end{enumerate}
\end{lem}

\begin{proof}
(1)~It follows from the definition that $\sqsubseteq$ is reflexive and transitive. If $x\sqsubseteq y\sqsubseteq x$, then for each $i$ there is a $j$ and then a $k$ with $[x_{2i},x_{2i+1}]\subseteq [y_{2j},y_{2j+1}]\subseteq [x_{2k},x_{2k+1}]$. But $i=k$, as otherwise these intervals would be disjoint. So each interval of $x$ is an interval of $y$ and vice versa. Thus $x=y$. So $\sqsubseteq$ is anti-symmetric and hence a partial order. Clearly $0$ and $1$ are the lower and upper bounds, respectively. (2)~The form of the definitions provides that the partial order of $\mc{K}(L)$ is the restriction of that of $\mc{K}^*(L)$. (3)~A recursive construction of the join of two elements of $\mc{K}(L)$ is given in \cite{HardingKalmbachMacNeille}. This construction and its proof shows that this is also their join in $\mc{K}^*(L)$. 
\end{proof}

By \cite{BalbesDwinger}*{p.~106}, there is a Boolean algebra embedding of the free Boolean extension $\mc{K}(C)$ of a bounded chain $C$ into the powerset $\mc{P}(C\setminus\{1\})$ taking the sequence $x$ to the union of half-open intervals $\bigcup\{[x_{2n},x_{2n+1}):n<\ell(x)\}$. We extend this to $\mc{K}^*(L)$ by allowing possibly infinite unions. 

\begin{prop}\label{bnm}
Suppose $C$ is an $(\omega+1)$-lattice that is a chain. Then there is an order-isomorphism $\Phi$ from $\mc{K}^*(C)$ to the powerset $\mc{P}(C\setminus\{1\})$ that preserves orthocomplementation:
\[\Phi(x) = \bigcup\{[x_{2n},x_{2n+1}):n<\ell(x) \}.\]
Thus $\mc{K}^*(C)$ is a complete atomic Boolean algebra and is equal to $\ol{\mc{K}(C)}$. 
\end{prop}

\begin{proof}
It follows from the definition of $\sqsubseteq$ that $\Phi$ is order preserving. To see that it is an order embedding, suppose $\Phi(x)\subseteq\Phi(y)$, and let $n<\ell(x)$. Then $x_{2n}\in\Phi(x)\subseteq\Phi(y)$ so there is some $m<\ell(y)$ with $x_{2n}$ belonging to the half-open interval $[y_{2m},y_{2m+1})$, hence with $y_{2m}\leq x_{2n}<y_{2m+1}$. Since $y_{2m+1}\not\in\Phi(y)$, it cannot be the case that $y_{2m+1}<x_{2n+1}$ since this would give $y_{2m+1}$ is in $[x_{2n},x_{2n+1})$ and hence in $\Phi(x)$. As $C$ is a chain, we have $x_{2n+1}\leq y_{2m+1}$. Therefore $y_{2m}\leq x_{2n}<x_{2n+1}\leq y_{2m+1}$. This shows that $x\sqsubseteq y$. So $\Phi$ is an order embedding. 

It is straightforward to show that $\Phi$ is onto and that $\Phi(x^\perp)$ is the set-theoretic complement of $\Phi(x)$. So $\Phi$ is an order-isomorphism preserving $\perp$. Since the powerset $\mc{P}(C\setminus\{1\})$ is a complete atomic Boolean algebra, so is $\mc{K}^*(C)$. 
Clearly $\mc{K}(C)$ is a Boolean subalgebra of $\mc{K}^*(C)$ and every atom of $\mc{K}^*(C)$ belongs to $\mc{K}(C)$. It follows that $\mc{K}^*(C)$ is the MacNeille completion of $\mc{K}(C)$. 
\end{proof}

For $x, y \in \K^*(L)$, let $x\cup y$ be the elements that appear as terms of $x$ or $y$, and for a family of sequences $(x^i)_{i \in I}$, let $\bigcup_{i \in I}x^i$ be the elements that appear as a term of $x^i$ for some $i\in I$.

\begin{cor}\label{ooo}
Suppose $C$ is an $(\omega+1)$-lattice that is a chain. If $(x^i)_{i \in I}$ is a family of elements in $\mc{K}^*(C)$ with join $x$, then each term in $x$ belongs to $\bigcup_{i \in I}x^i\cup\{0,1\}$. 
\end{cor}

\begin{proof}
Assume $C$ is the chain $c_0<c_1<\cdots <c_\alpha$ where $\alpha$ is finite or $\omega$. For $y\in\mc{K}^*(C)$ and $k<\alpha$, we have that 
\[
\begin{aligned}
 c_k = y_{2n} \mbox{ for some }n<\ell(y) & \mbox{ iff } c_k\in\Phi(y)\mbox{ and } c_{k-1}\not\in\Phi(y), \\
c_k = y_{2n+1}\mbox{ for some }n<\ell(y) & \mbox{ iff } c_k\not\in\Phi(y)\mbox{ and } c_{k-1}\in\Phi(y).
\end{aligned}
\]
Suppose $n<\ell(x)$ and $c_k = x_{2n}$. Then $c_k\in\Phi(x)$ and $c_{k-1}\not\in\Phi(x)$. Since $x=\bigvee_{i \in I}x^i$ and $\Phi$ is an isomorphism, we have $\Phi(x)=\bigcup_{i \in I}\Phi(x^i)$. So there exists an $i\in I$ with $c_k\in\Phi(x^i)$ and clearly $c_{k-1}\not\in\Phi(x^i)$. So $c_k$ is a term of $x^i$. Suppose that $c_p=x_{2n+1}$. Then $c_{p-1}\in\Phi(x)$ and $c_p\not\in\Phi(x)$. So there is $i\in I$ with $c_{p-1}\in\phi(x^i)$ and $c_p\not\in\phi(x^i)$. Again $c_p$ is a term of $x^i$. 
\end{proof}

\begin{cor}\label{zonk}
Let $C$ be an $(\omega+1)$-lattice that is a chain, and let $D$ be a bounded subchain of~$C$. Then the complete Boolean algebra $\mc{K}^*(D)$ is a complete subalgebra of $\mc{K}^*(C)$. 
\end{cor}

\begin{proof}
Let $(x^i)_{i \in I}$ be a family in $\mc{K}^*(D)$ and $x$ be its join in $\mc{K}^*(C)$. By Corollary~\ref{ooo} we have that $x$ is an element of $\mc{K}^*(D)$. Thus $x$ is the join of this family in $\mc{K}^*(D)$. 
\end{proof}

\begin{lem}\label{pout}
Let $L$ be an $(\omega+1)$-lattice and $(x^i)_{i \in I}$ be a family in $\mc{K}^*(L)$. If $D=\bigcup_{i \in I}x^i\cup\{0,1\}$ is a chain, then the join $x$ of $(x^i)_{i \in I}$ in $\mc{K}^*(D)$ is its join in $\mc{K}^*(L)$. 
\end{lem} 

\begin{proof}
From the form of their definitions, the partial order $\sqsubseteq$ in $\mc{K}^*(D)$ is the restriction of the partial order $\sqsubseteq$ in $\mc{K}^*(L)$. So $x$ is an upper bound of this family in $\mc{K}^*(L)$. Suppose $w$ is an upper bound of this family in $\mc{K}^*(L)$. For each $i\in I$, we have $x^i\sqsubseteq w$ and this implies that $x^i\cup w$ is a chain. 
Since $D$ and $w$ are chains and each element of $D$ is comparable to each element of $w$, we have that $C=D\cup w$ is a chain.
By Corollary~\ref{zonk}, the join $x$ of this family in $\mc{K}^*(D)$ is its join in $\mc{K}^*(C)$. But $w$ belongs to $\mc{K}^*(C)$ and is an upper bound of this family. Therefore $x\sqsubseteq w$. So $x$ is the join in $\mc{K}^*(L)$. 
\end{proof}

\begin{prop}\label{fredo}
If $L$ is an $(\omega+1)$-lattice, then $\mc{K}^*(L)$ is the MacNeille completion of $\mc{K}(L)$. 
\end{prop}

\begin{proof}
We use $\sqcup$ for the join of two elements in $\mc{K}^*(L)$ when it exists and $\bigsqcup$ for the join of an arbitrary family of elements when it exists. Note first that Lemma~\ref{pout} implies that any chain in $\mc{K}^*(L)$ has a join. 

Suppose $x, y\in\mc{K}^*(L)$. For each $n\in\mathbb{N}$ define $x^{n}$ as follows:
\[ x^{n} = \begin{cases} x_0<\cdots <x_{2n-1} & \mbox{ if } n\leq \ell(x) \\ x & \mbox{ otherwise } \end{cases} \]
Define $y^{n}$ similarly. It is easily seen that $x=\bigsqcup_\mathbb{N}x^n$ and $y=\bigsqcup_\mathbb{N}y^n$.
Set $z^{n} = x^{n}\sqcup y^{n}$. Note that this is defined since $x^{n}$ and $y^{n}$ are finite sequences  so they belong to $\mc{K}(L)$ and by Lemma~\ref{dvd} finite joins in $\mc{K}(L)$ are joins in $\mc{K}^*(L)$. For $m\leq n$, we have $x^{m}\sqsubseteq x^{n}$ and $y^{m}\sqsubseteq y^{n}$, hence $z^{m}\sqsubseteq z^{n}$. Let $z=\bigsqcup_\mathbb{N}z^n$ be the join of this chain in $\mc{K}^*(L)$. Since $x=\bigsqcup_\mathbb{N}x^n$ and $y=\bigsqcup_\mathbb{N}y^n$, it follows that $z=x\sqcup y$ in $\mc{K}^*(L)$. 

Any chain in $\mc{K}^*(L)$ has a join and any two elements of $\mc{K}^*(L)$ have a join. From general principles, it follows that $\mc{K}^*(L)$ is complete. Clearly $\perp$ is of period-two. If $x,y\in\mc{K}^*(L)$ and $x\sqsubseteq y$, then there is a bounded chain $C$ in $L$ with $x,y\in\mc{K}^*(C)$. The ordering and $\perp$ of $\mc{K}^*(C)$ are the restrictions of those of $\mc{K}^*(L)$. But $\mc{K}^*(C)$ is a Boolean algebra, so $x\sqcup x^\perp =1$ and $x\sqsubseteq y$ implies $y^\perp\sqsubseteq x^\perp$. It follows that $\perp$ is an orthocomplementation on $\mc{K}^*(L)$. 
As we have seen, each element $x\in\mc{K}^*(L)$ is the join $x=\bigsqcup_{\mathbb{N}}x^n$ of elements of $\mc{K}(L)$. So the ortholattice $\mc{K}(L)$ is join-dense in the complete ortholattice $\mc{K}^*(L)$, and therefore it is meet-dense as well. 
Being a join- and meet-dense subalgebra of a complete \ts{ol} characterizes the MacNeille completion up to isomorphism, so $\mc{K}^*(L)$ is the MacNeille completion of $\mc{K}(L)$. 
\end{proof}

\begin{prop}
If $L$ is an $(\omega+1)$-lattice, then $\mc{K}^*(L)$ is a complete \ts{oml}.
\end{prop}

\begin{proof}
Any chain in $L$ has a least upper bound in $L$ since it is either finite or it is infinite and therefore has $1$ as its join. Thus $L$ is a complete lattice. So by Theorem~\ref{cdf}, the MacNeille completion $\ol{\mc{K}(L)}$ is an \ts{oml}. By Proposition~\ref{fredo}, $\mc{K}^*(L)$ is the MacNeille completion of $\mc{K}(L)$ and therefore is a complete \ts{oml}. 
\end{proof}

Now that we know that $\mc{K}^*(L)$ is a complete \ts{oml}, we can state the following result that is a consequence of Corollary~\ref{zonk} and Lemma~\ref{pout}.

\begin{cor}\label{who}
If $L$ is an $(\omega+1)$-lattice and $C$ is a bounded subchain of $L$, then $\mc{K}^*(C)$ is a complete Boolean subalgeba of $\mc{K}^*(L)$.
\end{cor}

\begin{proof}
Let $(x^i)_{i \in I}$ be a family in $\mc{K}^*(C)$ and $x$ be the join of this family in the complete Boolean algebra $\mc{K}^*(C)$. Set $D=\bigcup_{i \in I}x^i\cup\{0,1\}$. By Corollary~\ref{zonk} we have that $x$ is the join of this family in $\mc{K}^*(D)$, and by Lemma~\ref{pout}, $x$ is the join of this family in $\mc{K}^*(L)$. So $\mc{K}^*(C)$ is closed under joins in $\mc{K}^*(L)$, and it is clearly closed under the operation $\perp$ of $\mc{K}^*(L)$ also.      
\end{proof}

For the following result we recall that the {\em commutator} of elements $x,y$ in an \ts{oml} is given by $\gamma(x,y)=(x\vee y)\wedge(x\vee y^\perp)\wedge(x^\perp\vee y)\wedge(x^\perp\vee y^\perp)$. It has the property that $\gamma(x,y)=0$ iff $x$ and $y$ commute. See \cite{KalmbachBook} for details. 

\begin{lem}\label{xdr}
If $L$ is an $(\omega+1)$-lattice and $x,y\in\mc{K}^*(L)$, then $x$ and $y$ commute iff $x\cup y$ is a chain. Thus the blocks of $\mc{K}^*(L)$ are the sets $\mc{K}^*(C)$ where $C$ is a maximal chain of $L$. 
\end{lem}

\begin{proof}
For $x\in\mc{K}^*(L)$ the terms of $x$ and $x^\perp$ are the same except possibly the bounds. Assume that $x \cup y$ is not a chain. Then, there are terms $u$ of $x$ and $v$ of $y$ that are incomparable. Clearly $u,v \not \in \{0,1\}$, and therefore $u$ is also a term of $x^\perp$ and $v$ is a term of $y^\perp$. Let $z$ be the sequence $u\wedge v < u\vee v$. We show that the sequence $z$ is below $\gamma (x,y)$ or equivalently below each of the sequences $x \sqcup y$, $x \sqcup y^\perp$, $x^\perp \sqcup y$, and $x^\perp \sqcup y^\perp$. Let $w = x \sqcup y$. Then there are $m,n<\ell(w)$ with $w_{2m}\leq u\leq w_{2m+1}$ and $w_{2n}\leq v\leq w_{2n+1}$. Since $u$ and $v$ are incomparable we must have $m=n$ and therefore there exists an $n<\ell(w)$ with $w_{2n}\leq u\wedge v<u\vee v\leq w_{2n+1}$. So $z\sqsubseteq w = x\sqcup y$. Similarly, $z \sqsubseteq x \vee y^\perp, x^\perp \vee y, x^\perp \vee y^\perp$. We conclude that $0 \neq z \sqsubseteq \gamma(x, y)$, that is, that $x$ and $y$ do not commute.

For the converse, if $x\cup y$ is a chain, it is contained in some bounded chain $C$ of $L$. Then $x,y$ are contained in the Boolean algebra $\mc{K}^*(C)$. By Corollary~\ref{who}, $\mc{K}^*(C)$ is a Boolean subalgebra of $\mc{K}^*(L)$. Thus $x$ and $y$ commute. 

For the further comment, suppose $B$ is a Boolean subalgebra of $\mc{K}^*(L)$. Then we have that $D=\bigcup\{x:x\in B\}$ is a chain in $L$. Extending $D$ to a maximal chain $C$ in $L$ we have $B\subseteq\mc{K}^*(C)$ and $\mc{K}^*(C)$ is a Boolean subalgebra of $\mc{K}^*(L)$. So if $B$ is a block, $B=\mc{K}^*(C)$ for some maximal chain $C$ of $L$. It remains to show that if $C$ is a maximal chain of $L$ then $\mc{K}^*(C)$ is a block. Suppose that $\mc{K}^*(C)$ is contained in a Boolean subalgebra $B$ of $\mc{K}^*(L)$. Then $B$ is contained in $\mc{K}^*(D)$ for some maximal chain $D$ in $L$. It is easily seen that $C=\bigcup\{x:x\in\mc{K}^*(C)\}$ and $D=\bigcup\{x:x\in\mc{K}^*(D)\}$. So $C\subseteq D$ and by the maximality of $C$ we have equality. Thus $\mc{K}^*(C)\subseteq B\subseteq \mc{K}^*(C)$, providing equality. So $\mc{K}^*(C)$ is a block. 
\end{proof}

\begin{thm}\label{ratso}
If $L$ is an $(\omega+1)$-lattice, then $\mc{K}^*(L)$ is a completely hereditarily atomic \ts{oml}.
\end{thm}

\begin{proof}
By Lemma~\ref{xdr}, the blocks of $\mc{K}^*(L)$ are exactly the sets $\mc{K}^*(C)$ where $C$ is a maximal chain in $L$. Since each maximal chain in $L$ is either finite or of order type $\omega+1$, each interval in a maximal chain in $L$ contains a cover, so $\mc{K}^*(C)$ is atomic for each maximal chain $C$ of $L$. Thus every block of $\mc{K}^*(L)$ is atomic, and this yields that $\mc{K}^*(L)$ is completely hereditarily atomic. 
\end{proof}

We next make use of this tool to construct an interesting example of a completely hereditarily atomic \ts{oml}. Consider the lattice $L$ shown at left in Figure~\ref{fig1}. This lattice is known as the Rieger-Nishimura lattice. 

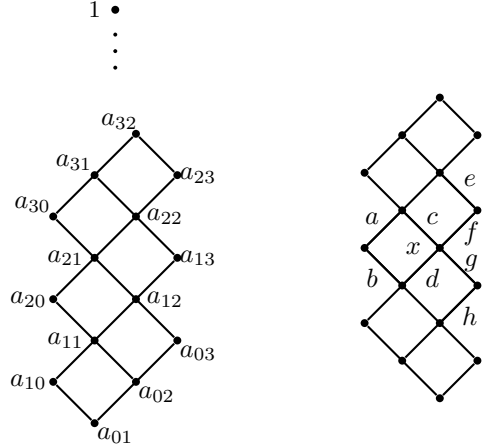
\begin{figure}
\begin{center}
\begin{tikzpicture}[scale=1.1,Q/.style={circle,fill=black,inner sep=0pt,minimum size=3pt},
R/.style={circle,fill=black,inner sep=0pt,minimum size=1.5pt}]
\node[Q] (01) at (0,0) {}; \node[Q] (02) at (.5,.5) {}; \node[Q] (03) at (1,1) {};  
\node[Q] (10) at (-.5,.5) {}; \node[Q] (11) at (0,1) {}; \node[Q] (12) at (.5,1.5) {}; \node[Q] (13) at (1,2) {}; 
\node[Q] (20) at (-.5,1.5) {}; \node[Q] (21) at (0,2) {}; \node[Q] (22) at (.5,2.5) {}; \node[Q] (23) at (1,3) {};
\node[Q] (30) at (-.5,2.5) {}; \node[Q] (31) at (0,3) {}; \node[Q] (32) at (.5,3.5) {}; 
\node[Q] (1) at (.25,5) {}; 
\draw[thick] (01) -- (03); \draw[thick] (10) -- (13); \draw[thick] (20) -- (23); \draw[thick] (30) -- (32); 
\draw[thick] (01) -- (10); \draw[thick] (02) -- (20); \draw[thick] (03) -- (30); \draw[thick] (13) -- (31); \draw[thick] (23) -- (32);
\node at (.25,-.15) {$a_{01}$}; \node at (.75,.35) {$a_{02}$}; \node at (1.25,.85) {$a_{03}$};
\node at (-.8,.5) {$a_{10}$}; \node at (-.35,1) {$a_{11}$}; \node at (.85,1.5) {$a_{12}$}; \node at (1.25,2) {$a_{13}$};
\node at (-.8,1.5) {$a_{20}$}; \node at (-.35,2) {$a_{21}$}; \node at (.85,2.5) {$a_{22}$}; \node at (1.25,3) {$a_{23}$};
\node at (-.75,2.65) {$a_{30}$}; \node at (-.25,3.15) {$a_{31}$}; \node at (.3,3.65) {$a_{32}$}; 
\node at (0,5) {\small{1}}; 
\node[R] at (.25,4.7) {}; \node[R] at (.25,4.5) {}; \node[R] at (.25,4.3) {};
\end{tikzpicture}
\hspace{10ex} 
\begin{tikzpicture}[scale=1.0,Q/.style={circle,fill=black,inner sep=0pt,minimum size=3pt},
R/.style={circle,fill=black,inner sep=0pt,minimum size=1.5pt}]
\node at (0,.9) {};

\draw[thick] (.5,2.5) -- (1,3) -- (0,4); \draw[thick] (0,3) -- (1,4) -- (.5,4.5); 
\draw[thick] (0,3) -- (-.5,3.5) -- (0,4); 

\node[Q] (12) at (.5,1.5) {}; \node[Q] (13) at (1,2) {}; 

\node[Q] (21) at (0,2) {}; \node[Q] (22) at (.5,2.5) {}; \node[Q] (23) at (1,3) {};
\node[Q] (30) at (-.5,2.5) {}; \node[Q] (31) at (0,3) {}; \node[Q] (32) at (.5,3.5) {}; \node[Q] (33) at (1,4) {};

\node[Q] (40) at (-.5,3.5) {}; \node[Q] (41) at (0,4) {}; \node[Q] (42) at (.5,4.5) {}; \node[Q] (43) at (1,5) {};
\node[Q] (50) at (-.5,4.5) {}; \node[Q] (51) at (0,5) {}; \node[Q] (52) at (.5,5.5) {}; 

\draw[thick] (12) -- (13); \draw[thick] (21) -- (23); \draw[thick] (30) -- (33); 
\draw[thick] (40) -- (43); \draw[thick] (50) -- (52); 
\draw[thick] (12) -- (30); \draw[thick] (13) -- (40); \draw[thick] (23) -- (50);
\draw[thick] (33) -- (51); \draw[thick] (43) -- (52); 

\node at (-.4,3.90) {$a$}; \node at (-.4,3.10) {$b$}; 
\node at (.4,3.90) {$c$}; 
\node at (.4,3.10) {$d$}; 
\node at (.15,3.5) {$x$}; 
\node at (0.9, 4.4) {$e$};
\node at (.92, 3.7) {$f$};
\node at (.92, 3.3) {$g$};
\node at (0.9, 2.6) {$h$};
\end{tikzpicture}
\end{center}
\caption{The Rieger-Nishimura lattice is the lattice depicted at left. At right is a portion of the Reiger-Nishimura lattice with edges of the Hasse diagram labelled $a,\ldots,h$ and an element of the lattice labelled $x$.} \label{fig1}
\end{figure}

\begin{defn}
An interval $[x,y]$ in a lattice is itself a lattice, and as such it has height at most $n$ if each chain in this interval has at most $n+1$ elements. An atomic lattice $L$ satisfies the {\em $n$-covering property} if for any atom $a$ and element $x\in L$ the interval $[x,x\vee a]$ has height at most $n$. 
\end{defn}

The 1-covering property is usually called the covering property. The covering property is closely related to modularity. Any modular atomic lattice has the covering property, and 
as seen in Proposition~\ref{dolph}, the algebraic \ts{oml}s with the covering property are exactly the products of finite-height modular \ts{ol}s. Our aim is to show for $L$ the Rieger-Nishimura lattice, that $\mc{K}^*(L)$ is a directly irreducible non-modular algebraic \ts{oml} with the 2-covering property. We proceed in stages. 

\begin{prop}
For $L$ the Reiger-Nishimura lattice, $\mc{K}^*(L)$ is a directly irreducible \ts{oml}. 
\end{prop}

\begin{proof}
Note that $L$ is an $\omega+1$ lattice, so by Theorem~\ref{ratso}, $\mc{K}^*(L)$ is an \ts{oml}. By Lemma~\ref{xdr}, the blocks of $\mc{K}^*(L)$ are the $\mc{K}^*(C)$ where $C$ is a maximal chain in $L$. Since 0,1 are the only elements of $L$ that belong to every maximal chain, the bounds of $\mc{K}^*(L)$ are the only elements that belong to every block. So the center of $\mc{K}^*(L)$ is trivial. This implies that $\mc{K}^*(L)$ is directly irreducible. 
\end{proof}

\begin{prop}
If $L$ is the Reiger-Nishimura lattice, then $\mc{K}^*(L)$ is algebraic. 
\end{prop}

\begin{proof}
By Theorem~\ref{ratso}, $\mc{K}^*(L)$ is a complete and atomic \ts{oml}. Since every atomic \ts{oml} is atomistic, it is enough to show that that atoms of $\mc{K}^*(L)$ are compact. For this, it is enough to show that if $p$ is an atom of $\mc{K}^*(L)$ and $S$ is a set of atoms that does not contain $p$ and satisfies $p\sqsubseteq \bigsqcup S$, then there is a finite $S'\subseteq S$ with $p\sqsubseteq \bigsqcup S'$. 

Referring to the right side of Figure~\ref{fig1}, we establish some terminology. For each element $x \in L$ that has degree four in the Hasse diagram of $L$, let $x_\downarrow: 0 < x$ and $x_\uparrow:x < 1$. They are complementary elements of $\K^*(L)$. The atoms of $\K^*(L)$ correspond to the edges of the Hasse diagram; an atom is below $x_\downarrow$ in $\K^*(L)$ iff both vertices of the corresponding edge are below $x$ in $L$, and an atom is below $x_\uparrow$ iff both vertices of its corresponding edge are above $x$ in $L$. There are 8 atoms labelled at right in Figure~\ref{fig1}, atoms $a,\ldots,h$. We call $a,b$ red atoms and $c,\ldots,f$ blue atoms. 
The red atoms are the only atoms of $\K^*(L)$ that are neither below $x_\downarrow$ nor below $x_\uparrow$. The blue atoms are the only atoms of $\K^*(L)$ that do not commute with the red atom $a$, and it is likewise for the red atom $b$. Indeed, the blocks of $\K^*(L)$ correspond to the maximal chains of $L$ by Lemma~\ref{xdr}, and the blue atoms correspond to exactly those edges of the Hasse diagram that cannot be in the same ascending path as the edge labelled $a$.

We classify the atoms of $\mc{K}^*(L)$ as internal atoms, external atoms, and exceptional atoms.  The internal atoms are those atoms such that both endpoints of their corresponding edge have degree four. The external atoms are those atoms such that one endpoint of their corresponding edge has degree two and the other has degree four, except the atom $a_{10} \cov a_{11}$. The exceptional atoms are the remaining atoms; they are $a_{01} \cov a_{02}$, $a_{02} \cov a_{03}$, $a_{01} \cov a_{10}$, $a_{02} \cov a_{11}$, and $a_{10} \cov a_{11}$.
\vspace{1ex}

\noindent {\em Claim 1:} Each internal atom is compact. 

\begin{proof}[Proof of claim.]
Let $c:x \cov y $ be an internal atom of $\K^*(L)$. To show that $c$ is compact, let $S$ be a set of atoms of $\K^*(L)$ such that $c\sqsubseteq\bigsqcup S$. Without loss of generality, $c \not \in S$. There are four atoms of $\mc{K}^*(L)$ that do not commute with $c$, the two red atoms $a$ and $b$ and the two blue atoms $e$ and $f$. All other atoms, besides $c$ itself, are orthogonal to $c$. So $S$ must contain at least one of these four atoms that do not commute with $c$, as otherwise $\bigsqcup S$ would be orthogonal to $c$. Note that the join of any two of these four atoms lies above $c$ in $\K^*(L)$, so if $S$ contains at least two of these four atoms, then trivially there is a finite subset $S' \subseteq S$ such that $c \sqsubseteq \bigsqcup S'$. We are left then only to consider the case where exactly one of these four atoms belongs to $S$.

Assume that $a$ is the unique atom in $S$ that does not commute with $c$. The blue atoms $c,\ldots,h$ are exactly the atoms that do not commute with $a$. Assume that $S$ does not contain a blue atom, and let $T=S\setminus\{a\}$. Since $c\sqsubseteq\bigsqcup S$, we have that $c = c\sqcap \big(a\sqcup \bigsqcup T\big)$. Since each element of $T$ commutes with both $a$ and $c$, we also have that $\bigsqcup T$ commutes with both $a$ and $c$. By the Foulis-Holland Theorem \cites{BrunsHarding,KalmbachBook}, we may distribute $c$ to obtain that $c=(c\sqcap a)\sqcup \big( c\sqcap \bigsqcup T\big)$. This is a contradiction because $a\sqcap c=0$ and every element of $T$ is orthogonal to $c$. Therefore, $S$ does contain a blue atom. The join of $a$ with any of the blue atoms lies above $c$, so we obtain a finite subset of $S$ whose join is above $c$.

We have shown that if $a$ is the unique atom in $S$ that does not commute with $c$, then $S$ also contains one of the blue atoms $c,\ldots, h$ whose join with $a$ lies above $c$. If $b$ is the unique atom in $S$ that does not commute with $c$, then the we may apply the same argument. If $e$ (resp.~$f$) is the unique atom in $S$ that does not commute with $c$, then we may apply the same argument to obtain that $S$ must contain an atom that does not commute with $e$ (resp.~$f$) and the join of this atom with $e$ (resp.~$f$) lies above $c$. Therefore, $c$ is compact.
\end{proof}

\noindent {\em Claim 2:} Each external atom is compact. 

\begin{proof}[Proof of claim.]
Let $a : w \cov y$ be an external atom of $\K^*(L)$, where $w$ is a vertex of degree two and $y$ is a vertex of degree four. Let $x$ be the predecessor of $y$ that is distinct from $w$; it is a vertex of degree four, as depicted in Figure~\ref{fig1}. To show that $a$ is compact, let $S$ be a set of atoms of $\K^*(L)$ such that $a \sqsubseteq S$. Without loss of generality, $a \not \in S$. The blue atoms $c, \ldots, h$ are exactly the atoms that do not commute with $a$. All other atoms, besides $a$ itself, are orthogonal to $a$, so $S$ must contain at least one blue atom. The join of $b$ with any blue atom lies above $a$, so if $b \in S$, then $S$ contains two atoms whose join lies above $b$.

It remains to consider the case $a, b \not \in S$. Let $p, q \in \K^*(L)$ be given by
$$p = \bigsqcup \{ s \in S : s \sqsubseteq x_{\downarrow}\}, \qquad \qquad q = \bigsqcup \{s \in S : s \sqsubseteq x_\uparrow\}.$$
Since $s \sqsubseteq x_\downarrow$ or $s \sqsubseteq x_\uparrow$ for each $s \in S$, we have that $a \sqsubseteq \bigsqcup S = p \sqcup q$. We reason that
$
c \sqsubseteq d \sqcup a \sqsubseteq x_\downarrow \sqcup p \sqcup q = x_\downarrow \sqcup q$ and $
d \sqsubseteq a \sqcup c \sqsubseteq p\sqcup q\sqcup x_\uparrow = p\sqcup x_\uparrow.
$
The sequences $x_\downarrow, x_\uparrow \in \K^*(L)$ are orthogonal. Thus, $c, q \sqsubseteq x_\uparrow$ implies that $c \sqsubseteq q$, and $d, p \sqsubseteq x_\downarrow$ implies that $d \sqsubseteq p$. Altogether, we have that $c, d \sqsubseteq p \sqcup q = \bigsqcup S$. 

The atoms $c$ and $d$ are internal, so by claim 1, they are compact. In this way, we obtain a finite subset $S' \subseteq S$ such that $c,d \sqsubseteq \bigsqcup S'$. We conclude that $a \sqsubseteq c \sqcup d \sqsubseteq \bigsqcup S'$. Thus, $a$ is compact. We may apply the same argument to show that each external atom $b: z \cov w$ such that $w$ is a vertex of degree two and $z$ is a vertex of degree four is compact.
\end{proof}

We have a lattice embedding $L \to L$ that is given by $a_{ij} \mapsto a_{(i+1)j}$ and $1 \mapsto 1$. It induces a lattice embedding $\K^*(L) \to \K^*(L)$ that maps each atom to an internal atom or an external atom. It follows immediately that each atom of $\K^*(L)$ is compact. Therefore, $\mc{K}^*(L)$ is algebraic. 
\end{proof}

\begin{prop}
For $L$ the Reiger-Nishimura lattice, $\mc{K}^*(L)$ has the 2-covering property. 
\end{prop}

\begin{proof}
Let $x,y\in\mc{K}^*(L)$ with $x$ being an atom. We must show the interval $[y,x\sqcup y]$ in $\mc{K}^*(L)$ has height at most two. Since this interval is isomorphic to the interval $[0,(x\sqcup y)\sqcap y^\perp]$ \cite{BrunsHarding}*{Sec.~2}, it is enough to show that $[0,(x\sqcup y)\sqcap y^\perp]$ has height at most two. As is customary, we refer to the height of an interval $[0,z]$ as the \emph{height} of $z$. Thus we must show that $(x\sqcup y)\sqcap y^\perp$ has height at most two. This is obvious when $x\sqsubseteq y$ or $y=0$, so we assume $x\not\sqsubseteq y$ and $y\neq 0$. The equality $(x\sqcup y)\sqcap y^\perp =  x \sqcup y $ would imply that $y\sqsubseteq y^\perp$, so we have that $(x\sqcup y)\sqcap y^\perp \sqsubset x \sqcup y$.

Since $\mc{K}^*(L)$ is atomic, $y$ is equal to the join of a family $S$ of pairwise orthogonal atoms. Let $T$ be the atoms in $S$ that commute with $x$ and $U$ be the atoms in $S$ that do not commute with $x$. Since $x\not\sqsubseteq y$, each member of $T$ is orthogonal to $x$. We compute that
\[
(x\sqcup y)\sqcap y^\perp \,\,=\,\, 
\left(x\sqcup\bigsqcup U\sqcup\bigsqcup T\right)\sqcap \left(\bigsqcup T\right)^\perp\sqcap \left(\bigsqcup U\right)^\perp 
\,\,=\,\, \left(x\sqcup \bigsqcup U\right)\sqcap \left(\bigsqcup U\right)^\perp.
\]

\noindent For the first equality we use $y=\bigsqcup S = \bigsqcup T\sqcup\bigsqcup U$. For the second, we use the orthomodular law. Indeed, each member of $T$ is orthogonal to $x$ and to each member of $U$, so we have that $x\sqcup\bigsqcup U\leq (\bigsqcup T)^\perp$. Hence, to verify the 2-covering property, we may assume that $y$ is a join of atoms that do not commute with $x$. 

Assume that $x$ is an internal atom. Without loss of generality, we can assume it is the atom $c$ in Figure~\ref{fig1}. Then there are four atoms that do not commute with $x=c$: the two red atoms $a$ and $b$ and the two blue atoms $e$ and $f$. Since the join of any two of these four atoms lies above $x=c$ and we have assumed that $x\not\sqsubseteq  y$, it must be that $y$ is one of these four atoms. In each case, $(x\sqcup y)\sqcap y^\perp$ is an atom and hence has height one. 

Assume that $x$ is an external atom, without loss of generality the atom $a$ in Figure~\ref{fig1}. The atoms that do not commute with $x=a$ are the blue atoms $c, \ldots, h$. Hence $y$ is a join of an orthogonal family of blue atoms. We have that $a, \ldots, h \sqsubseteq e \sqcup f \sqcup g \sqcup h$, so $x \sqcup y \sqsubseteq e \sqcup f \sqcup g \sqcup h$. If this inequality is strict, then $x \sqcup y \sqsubset e \sqcup f \sqcup g \sqcup h$ has height at most three and hence $(x\sqcup y)\sqcap y^\perp \sqsubset x \sqcup y$ has height at most two. If instead $x \sqcup y = e \sqcup f \sqcup g \sqcup h$, then $(x \sqcup y) \sqcap y^\perp$ has height at most two, because $y$ must be the join of at least two blue atoms and $\K^*(L)$ has no elements of height greater than two that are both below $x \sqcup y$ and that are orthogonal to any two blue atoms.

Each of the five exceptional atoms may be treated by an appropriate variant of one of the two above arguments. 
\end{proof}

\begin{thm}
Let $L$ be the Rieger-Nishimura lattice depicted in Figure~\ref{fig1}. Then, $\mc{K}^*(L)$ is a directly irreducible, algebraic \ts{oml} with the 2-covering property that has infinite height.
\end{thm}

\section{Keller's example}

In this section, we provide an example of a directly irreducible, completely hereditarily atomic \ts{oml} with the covering property. This \ts{oml} is built by a process that is similar to the construction of $\mc{C}(\mc{H})$ for $\mc{H}=\ell^2$, the space of square summable real-valued sequences, but with sequences taking values in a particular non-archimedean ordered field $\KK$ rather than in the reals. This ingenious construction was pioneered by Keller \cite{Keller}. See \cite{GrossKunzi} for a survey. Our basic ingredient is a {\em Hermitian space}, a vector space $\E$ over an involutive division ring $\KK$ together with an anisotropic Hermitian form $\langle \cdot |\cdot \rangle$ on $\E$.

\begin{thm}\label{thm: representation}
If $\E$ is a Hermitian space, then its \ts{ol} $\mc{L}_{\perp\perp}$ of biorthogonally closed subspaces is directly irreducible, complete, atomistic, and has the covering property. It is an \ts{oml} iff $\E$ is an {\em orthomodular space}, meaning $\E=X+X^\perp$ for each $X\in\mc{L}_{\perp\perp}$. Conversely, if $L$ is a directly irreducible, complete, atomistic $\ts{ol}$ with the covering property that is of height at least 4, then there is a Hermitian space $\E$ so that $L$ is isomorphic to the \ts{ol} $\mc{L}_{\perp\perp}$. 
\end{thm}

This theorem combines \cite{MaedaMaeda}*{Thm.~34.2, 34.5} and \cite{GrossKunzi}*{p.~191}.  
It situates the work of Gross, Keller, K\"unzi and others. It also provides an alternative proof that any directly irreducible, algebraic \ts{oml} with the covering property has finite height:

\begin{thm}[i.e., \ref{predolph}]
Let $L$ be a directly irreducible, algebraic \ts{oml} with the covering property. Then, $L$ is modular and has finite height.
\end{thm}

\begin{proof}
If $L$ is of height strictly less than four, then $L$ is modular because it is straightforward to show that $L$ does not contain the pentagon as a sublattice. Hence we assume that $L$ is of height at least four. Thus, there is a Hermitian space $\E$ such that $L$ is isomorphic to the ortholattice $\mc{L}_{\perp\perp}$ of biorthogonally closed subsets of $\E$. If $\E$ is finite-dimensional, the result is clear. Toward a contradiction, assume that $\E$ is not finite-dimensional. Then $\E$ has a maximal orthogonal set $F = \{f_\alpha\}_{\alpha \in \beta}$ where $\beta$ is an infinite ordinal. Finite-dimensional subspaces of $\E$ are biorthogonally closed \cite{MaedaMaeda}*{Lem.~33.3}, and as $L$ is algebraic, they are compact elements of $\mc{L}_{\perp\perp}$ . This implies that each atom $\KK f$ is under a finite join of the atoms $\KK f_i$, so $F$ is a basis for $\E$. We now observe that the $G = \{f_\alpha - f_{\alpha+1}\}_{\alpha \in \beta}$ has the property that $G^{\perp\perp} = \E$, so $G$ is a basis for $\E$. This is absurd since $f_0$ is not in the span of $G$. Therefore, $\E$ has finite height and is modular.
\end{proof}

We next construct the Hermitian space whose biorthogonal sets will provide the example desired in this section. It is described in more detail in \cite{Keller-Ochsenius-95}*{Sec.~4}. Let $\Gamma$ be the abelian group $\ZZ \oplus \ZZ \oplus \cdots$ with the  \emph{reverse lexicographic order}: $\gamma < \gamma'$ if $\gamma(n) < \gamma'(n)$ for the largest $n$ such that $\gamma(n) \neq \gamma'(n)$. Let $\KK$ be the field $\RR[[t^\Gamma]]$ of generalized power series, i.e., Hahn series, with exponents in $\Gamma$ and coefficients in $\RR$. So elements of $\KK$ are functions $x\:\Gamma\to\RR$ whose support is a well-ordered subset of $\Gamma$. Equip $\KK$ with the trivial involution and the valuation $\varphi:\KK\to\Gamma\cup\{\infty\}$ where $\varphi(x)=\min\{\gamma:x(\gamma)\neq 0\}$, for $x \in \KK \setminus \{0\}$, and $\varphi(0) = \infty$. As with all valued fields, $\KK$ carries a topology, where $x_n\to x$ if for each $\gamma\in\Gamma$ there is $N\in\NN$ with $\varphi(x-x_n)>\gamma$ for all $n>N$. 

\begin{thm}[\cite{Keller-Ochsenius-95}*{Sec.~4}] 
The valued field $(\KK,\varphi)$ is complete. 
\end{thm}

For $n\in\NN$, let $\delta_n \in \Gamma$ be the sequence with $\delta_n(n)=1$ and $\delta_n(k)=0$ for $k\neq n$, and let $t_n\in\KK$ be the function $t_n:\Gamma\to\RR$ with $t_n(\delta_n)=1$ and $t_n(\gamma)=0$ for $\gamma\neq\delta_n$. Note $\varphi(t_n)=\delta_n$. Let $\E$ be the set of all sequences $f$ of elements of $\KK$ that are ``square-summable'' in the sense that $\sum_\NN f(n)^2t_n$ converges in the valuation topology, or equivalently with $f_n^2t_n\to 0$. Define $\<\cdot,\cdot\>:\E\times\E\to\KK $ by 
$$\< f,g\> = \sum_\NN f(n)g(n)t_n.$$
For convenience, we write $\< f\>$ for $\< f,f\>$. For $n\in\NN$, let $e_n$ be the sequence whose $n^{th}$ term is 1 and is 0 otherwise. Note that $\{e_n\}_{n \in \NN}$ is a maximal orthogonal family in $\E$ and $\langle e_n\rangle = t_n$. Finally, let  $\mc{L}$ be the \ts{ol} $\mc{L}_{\perp\perp}$ of biorthogonally closed subsets of $\E$. 

\begin{thm}\label{much}
$\mc{L}$ is an infinite height, directly irreducible, complete, atomic \ts{oml} with the covering property. 
\end{thm}

\begin{proof}
$\mc{L}$ is of infinite height since $\E$ is infinite-dimensional. By \cite{Keller-Ochsenius-95}*{Thm.~4}, $\E$ is an orthomodular space. This means that $\E$ is a Hermitian space, so by Theorem~\ref{thm: representation}, $\mc{L}$ is directly irreducible, complete, atomic, and has the covering property, and additionally that $\mc{L}$ is an \ts{oml} \cite{GrossKunzi}*{p.~191}. 
\end{proof}

To establish that all the blocks of $\mc{L}$ are atomic, we require further properties of $\E$. Note first that $\KK$ satisfies the ``triangle inequality'' $\varphi(\< f+g\>)\geq\min\{\varphi(\< f\>),\varphi(\< g\>)\}$ \cite{GrossKunzi}*{p.~195}, so $\E$ is a {\em definite space} in the sense of \cite{GrossKunzi}*{p.~196}. Let  $T:\Gamma\to\Gamma/2\Gamma=\ZZ_2\oplus\ZZ_2\oplus\cdots$ be the natural quotient map, and for $f\in \E$ define its {\em type} to be $\type(f)=\varphi(\langle f\rangle)$. Note that $\type(e_n)=T(\delta_n)$, so the $e_n$ each have different types. We next establish a result that says $\{e_n\}_{n \in \NN}$ satisfies the ``type condition''. This result is surely known, but we prove it since we cannot find the explicit statement. 

\begin{lem}\label{bnq}
If $(x_n)_{n \in \NN}$ is a family in $\KK$ such that $\{\varphi(\< x_ne_n\>):n\in\NN\}$ is bounded below, then $x_ne_n\to 0$ and hence $\sum_\NN x_ne_n$ converges.   
\end{lem}

\begin{proof}
Suppose $\varphi(\< x_ne_n\>)>2\gamma$ for $n\in\NN$, and set $\gamma_n=\varphi(\< e_n\>)+2\varphi(x_n)-2\gamma$. Note that $\gamma_n>0$ and $T(\gamma_n)=\type(e_n)=T(\delta_n)$. This implies $\gamma_n>\delta_{n-1}$. So for each $\gamma'\in\Gamma$, all but finitely many $n$ satisfy $\gamma_n > \gamma'$, so $\gamma_n+2\gamma\to\infty$. This gives $\varphi(\langle x_ne_n\rangle)\to\infty$, hence that $x_ne_n\to 0$.
\end{proof}

Therefore $\E$ is a definite space that has a maximal orthogonal set $\{e_n\}_{n \in \NN}$ satisfying Lemma~\ref{bnq} known as the {\em type condition}. The counting types function $v$ for $\{e_n\}_{n\in\NN}$ is the map $v:\Gamma/2\Gamma \to \NN$ where $v(t) = \mbox{card}\{n:\type(e_n)=t\}$ \cite{GrossKunzi}*{p.~199}. Since the $e_n$ have different types, this function takes values in $\{0,1\}$. Let $\Delta=\{\delta_n:n\in\NN\}$ and note that $\{\type(e_n):n\in\NN\}=T[\Delta]$. 

\begin{lem}\label{foul} $ $
\begin{enumerate}
\item Orthogonal vectors in $\E$ have different types, 
\item If $S$ is a maximal orthogonal set in $\E$, then $\type[S]=T[\Delta]$, 
\item If $X\in\mc{L}$ and $P$ and $Q$ are maximal orthogonal subsets of $X$, then $\type[P]=\type[Q]$.  
\end{enumerate}
\end{lem}

\begin{proof}
The first statement is given by \cite{GrossKunzi}*{Lem.~25} and the second by \cite{GrossKunzi}*{Cor.~26}. For the third, let $R$ be a maximal orthogonal set in $X^\perp$. Then the first two statements give that $R$ is a set-theoretic complement of both $P$ and $Q$ in $T[\Delta]$.
\end{proof}

\begin{defn}\label{Pi}
For each $X \in \mc{L}$, let $\pi(X)= \type[P]$, where $P$ is any maximal orthogonal subset $P \subsetof X$. This defines a function $\pi$ from $\mc{L}$ to the powerset $\mc{P}(T[\Delta])$.
\end{defn}

\begin{lem}\label{dd}
The function $\pi:\mc{L}\to \mc{P}(T[\Delta])$ has the following properties:
\begin{enumerate}
\item $\pi(X) = \emptyset$ iff $X = \{0\}$,
\item $\pi(X) \leq \pi(Y)$ whenever $X \subseteq Y$,
\item $\pi(X \vee Y) = \pi(X) \cup \pi(Y)$ whenever $X \perp Y$,
\item $\pi(X^\perp) = T[\Delta] \setminus \pi(X)$,
\item $\pi\big(\bigvee_{\alpha \in \kappa} X_\alpha\big) = \bigcup_{\alpha \in \kappa} \pi(X_\alpha)$ for each increasing sequence $\{ X_\alpha:\alpha\in\kappa\}$, where $\kappa$ is an ordinal.
\end{enumerate}
Therefore, the restriction of $\pi$ to each block $\B \subseteq \mc{L}$ is a complete homomorphism that is injective. 
\end{lem}

\begin{proof}
To prove (1), we observe that $\{0\}$ is the unique element of $\mc{L}$ for which the empty set is a maximal orthogonal subset. To prove (2), we observe that any maximal orthogonal subset of $X$ can be extended to a maximal orthogonal subset of $Y$. To prove (3), we observe if $P$ and $Q$ are maximal orthogonal subsets of $X$ and $Y$, respectively, then $P \cup Q$ is a maximal orthogonal subset of $X \vee Y$. To prove (4), we observe that any maximal orthogonal subset $P \subseteq X$ can be extended to a maximal orthogonal subset $R \subseteq \E$ and that the restricted function $\type|_R\: R \to T[\Delta]$ is a bijection by Lemma~\ref{foul}. We have proved (1)--(4).

To prove (5), let $\{ X_\alpha:\alpha\in\kappa\}$ be an increasing well-ordered sequence. Using transfinite recursion, we construct an increasing sequence $\{P_\alpha:\alpha\in\kappa\}$ of sets such that $P_\alpha$ is a maximal orthogonal subset of $X_\alpha$ for each $\alpha \in \kappa$. 
Then $ P = \bigcup\{P_\alpha:\alpha\in\kappa\}$ is a maximal orthogonal subset of the subspace $X = \bigvee\{X_\alpha:\alpha\in\kappa\}$ because any vector $f$ that is orthogonal to $P$ is immediately orthogonal to $X_\alpha$ for each $\alpha \in \kappa$ and hence to $X$ itself. Finally, we compute that $\pi(X) = \type[P] = \type[\bigcup_{\alpha \in \kappa} P_\alpha] = \bigcup_{\alpha \in \kappa} \type[P_\alpha] = \bigcup_{\alpha \in \kappa} \pi(X_\alpha)$, establishing (5).

Let $\B$ be a block of $\mc{L}$. The restriction of $\pi$ to the block $\B$ is an injective homomorphism $\B \to \mc{P}(T[\Delta])$ because $\pi$ has properties (1)--(4). This restriction is a complete homomorphism because $\pi$ has property (5).
\end{proof}

\begin{thm}
$\mc{L}$ is an infinite height, directly irreducible, completely hereditarily atomic \ts{oml} with the covering property. 
\end{thm}

\begin{proof}
In view of Theorem~\ref{much}, it remains only to show that each block $\mc{B}$ of $\mc{L}$ is atomic. By Lemma~\ref{dd} $\pi|_\B:\mc{B}\to\mc{P}(T[\Delta])$ is an injective complete homomorphism, hence, its image $\pi[\B]$ is isomorphic to $\B$ and is a complete subalgebra of $\mc{P}(T[\Delta])$. Since $\mc{P}(T[\Delta])$ is completely distributive, so is the complete subalgebra $\pi[\mc{B}]$. By Tarski's theorem, this implies that $\pi[\mc{B}]$ is atomic and therefore that $\mc{B}$ is atomic. 
\end{proof}

\begin{rem}
Each of the blocks of $\mc{L}$ has countably many atoms. The only known non-separable orthomodular spaces are Hilbert spaces \cite{GrossKunzi}*{Problem~2}.
\end{rem}

\section*{Acknowledgement}

The second author thanks Wolfgang Rump for suggesting Keller's Hermitian space as a potential example.

\vspace{5ex}

\noindent 
John Harding\\
New Mexico State University, Las Cruces NM 88003\\
\text{jharding@nmsu.edu}
\vspace{2ex}

\noindent
Andre Kornell\\
Dalhousie University, Halifax NS B3H 4R2\\
\text{akornell@dal.ca}

\end{document}